\documentclass[acmsmall,nonacm]{acmart}
\usepackage[utf8]{inputenc}

\usepackage{algorithm}
\usepackage[noend]{algpseudocode}
\usepackage{bbm}
\usepackage{thm-restate}
\usepackage{graphicx, color, xcolor}

\usepackage{relsize}
\usepackage{float}
\usepackage{enumitem}

\algnewcommand{\LineComment}[1]{\State \(\triangleright\) #1}
\newcommand{\comm}{\textsf{{Comm}}}

\newtheorem{definition}{Definition}

\newtheorem{theorem}{Theorem}
\newtheorem{lemma}{Lemma}
\newtheorem{corollary}{Corollary}

\begin{document}

\title{Improved Byzantine Agreement under an Adaptive Adversary}

\author{Fabien Dufoulon}
\orcid{0000-0003-2977-4109}
\affiliation{
\institution{Lancaster University}
\city{Lancaster}
\country{UK}
}
\email{f.dufoulon@lancaster.ac.uk}

\author{Gopal Pandurangan}
\orcid{0000-0001-5833-6592}
\affiliation{
  \institution{ University of Houston}
  \city{Houston, TX 77204}
  \country{USA}
}
\email{gopal@cs.uh.edu}

\date{}

\thispagestyle{empty}

\begin{abstract}
Byzantine agreement is a fundamental problem in fault-tolerant distributed computing that has been studied intensively for the last four decades. Much of the research has focused on a static Byzantine adversary, where the adversary is constrained
to choose the Byzantine nodes in advance of the protocol's execution. This work focuses on the harder case of an adaptive Byzantine adversary that can choose the Byzantine nodes \emph{adaptively} based on the protocol's execution.
While efficient $O(\log n)$-round protocols ($n$ is the total number of nodes) are known for the static adversary (Goldwasser, Pavlov, and Vaikuntanathan, FOCS 2006) tolerating up to $t < n/(3+\epsilon)$ Byzantine nodes, $\Omega(t/\sqrt{n \log n})$ rounds is a well-known lower bound for adaptive adversary [Bar-Joseph and Ben-Or, PODC 1998]. The best-known protocol for adaptive adversary runs in $O(t/\log n)$ rounds [Chor and Coan, IEEE Trans. Soft. Engg., 1985]. 

This work presents a synchronous randomized Byzantine agreement protocol under an adaptive adversary that improves over previous results. Our protocol works under the powerful {\em adaptive rushing adversary in the full information model}.  That is, we assume that the Byzantine nodes can behave arbitrarily and maliciously, have knowledge about the entire state of the network at every round, including random choices made by all the nodes up to and including the current round, have unlimited computational power, and may collude among themselves.  Furthermore, the adversary can {\em adaptively} corrupt up to $t < n/3$ nodes based on the protocol's execution.
We present a simple randomized Byzantine agreement protocol that runs in $O(\min\{t^2\log n/n, t/\log n\})$ rounds  that improves over the  long-standing bound of $O(t/\log n)$ rounds due to Chor and Coan [IEEE Trans. Soft. Engg., 1985]. Our bound is significantly better than that of Chor and Coan for $t = o(n/\log^2 n)$ and approaches the Bar-Joseph and Ben-Or [PODC 1998] lower bound of $\Omega(t/\sqrt{n\log n})$ rounds when $t$ approaches $\sqrt{n}$. 
\end{abstract}

\maketitle

\section{Introduction}
\label{sec:introduction}
Byzantine Agreement has been a central problem in distributed computing since it was introduced by the seminal work of  Pease, Shostak and Lamport \cite{Pease_1980}  in the early 1980s, and has been studied intensively for five decades, see e.g., 
 \cite{lynch, attiya, Dolev_1982, Rabin_1983, Ben-Or_1983, CC85,  Feldman_1997, Goldwasser_2006, Ben-Or_2006, King_2006_SODA, Kapron_2010, King_2011, Augustine_2020_DISC, saia, pettie}. The Byzantine agreement problem can be stated as follows.

 \begin{definition}[Byzantine agreement]
Let $P$ be a protocol on a distributed network of $n$ nodes in which each node $v$ starts with an input bit value $b_v$. A Byzantine adversary controls up to $t$ nodes, called Byzantine (or faulty), which can deviate arbitrarily from $P$. Protocol $P$ solves Byzantine agreement  if each (honest) node $v$ running $P$ \emph{terminates} and  outputs a value $o_v$ at the end of $P$ such that:
\begin{description}
\item [Agreement:] For any two honest nodes $u$ and $v$, $o_u = o_v$.
\item [Validity:] If the input value for all (honest) nodes is $b$, then the output value for all honest nodes
should be $b$.
\end{description}
\end{definition}

\noindent Byzantine agreement in distributed (message-passing) networks\footnote{As is standard in all prior works discussed in this paper, we assume point-to-point communication between all pairs of nodes, i.e., a complete network of $n$ nodes (cf. Section \ref{sec:model}).} has been studied extensively under various settings  (see also Section \ref{sec:related}). Some of the important ones are as follows: 
\begin{description}
\item [Deterministic versus Randomized Protocols:]  The protocol can be {\em deterministic} or {\em randomized} (where nodes have access to private or shared random bits). In randomized protocols (considered in this paper), agreement and/or performance (time, communication) guarantees are probabilistic.
\item [Full Information versus Weaker Adversarial Models:] In the setting of  {\em private channels} it is assumed that players can communicate using pairwise private communication channels which are not known to the Byzantine nodes. In a {\em computationally bounded} model the Byzantine nodes are assumed to be computationally bounded, and cryptographic primitives are assumed to exist. Finally, in the {\em full information model} (which is the focus of this paper), which is the most powerful adversarial model, no assumptions are made on the existence of private channels, nor are Byzantine nodes computationally bounded. Furthermore, in the full information model for randomized protocols, it is typically assumed that  Byzantine nodes know the random choices of all the honest nodes made till the previous round (and not future rounds). If they also know the random choices of the honest nodes in the \emph{current round} before they act, then
the adversary is called {\em rushing}. In other words, a rushing adversary in round $r$ can act based on random choices made by all honest nodes till round $r$, whereas a non-rushing can act only based on random choices made by all honest nodes till round $r-1$. 
\item [Static versus Adaptive Byzantine Adversary:] {\em Static} adversary, which is constrained to choose the Byzantine nodes in advance of the protocol's execution, and {\em adaptive} adversary, where the Byzantine nodes can be chosen
during the execution of the protocol based on the states of the nodes at any time.
\item [Synchronous versus Asynchronous Communication:] The underlying network communication model can be {\em synchronous} (as assumed here) or {\em asynchronous}.
\end{description}

We note that many works have assumed a {\em static} Byzantine adversary (discussed below). Our current work focuses on the significantly harder case of {\em adaptive} and {\em rushing} Byzantine adversary under the {\em full information model} in the {\em synchronous} setting.\footnote{Under the same assumptions, the asynchronous setting is even harder, and getting even polynomial (in $n$) round algorithms is hard --- cf. Section \ref{sec:related}.} Note that the full-information model has received much attention since protocols that work under this model \emph{do not rely on cryptographic assumptions} and can work without computational restriction assumptions. We refer to the   work of Chor and Coan \cite{CC85} for a discussion on early work on the Byzantine agreement in the full information model versus otherwise,  static versus adaptive adversary, deterministic versus randomized protocols, and asynchronous versus synchronous communication.  We note that for \emph{deterministic} protocols, there is a well-known lower bound of $t+1$ rounds (where $t$ is the number of Byzantine nodes) \cite{FischerL82} and $O(t)$-round deterministic protocols are known \cite{lamport82,Dolev_1982,garay}. For these protocols, while that of \cite{lamport82} has exponential communication complexity, the protocols of \cite{Dolev_1982,garay,kowalskibyz} have polynomial communication.

There is no difference between a static and adaptive adversary for deterministic protocols, as the entire execution is determined at the beginning of the protocol. However, the difference is essential for \emph{randomized} protocols: the \emph{static} adversary is oblivious to the random choices made during the protocol's execution; in contrast, the \emph{adaptive} adversary can choose to corrupt nodes
based on random choices till the current round. It is important to note that in the full information model, the \emph{maximum} number of Byzantine nodes that can be tolerated is $t < n/3$, even for a static adversary in the synchronous setting and even for randomized protocols \cite{FLM,yao}. On the other hand, (even) deterministic protocols that tolerate up to this limit are known (e.g., \cite{lamport82,Dolev_1982}).

For the {\em static} adversary, after a long line of research (see e.g., \cite{Goldwasser_2006, Ben-Or_2006} and the references therein) very efficient Byzantine agreement protocols --- taking  $O(\log n)$-rounds --- are known tolerating up to near-optimal bound of $t < n/(3+\epsilon)$ Byzantine nodes \cite{Goldwasser_2006} in the {\em full-information} model. In particular, the 2006 work of Goldwasser, Pavlov, and Vaikuntanathan \cite{Ben-Or_2006} points out  that their $O(\log n)$-round randomized protocol is a significant improvement over the 1985 work of Chor and Coan \cite{CC85} which presented a $O(t/\log n)$ round {\em randomized} Byzantine agreement protocol (which itself was an improvement over the $O(t)$-round deterministic protocol of \cite{Dolev_1982}). However, it must be pointed out that while Chor and Coan works under an {\em adaptive} adversary (though non-rushing\footnote{It is easy to make Chor and Coan's protocol work under a \emph{rushing} adaptive adversary, using an idea similar to our protocol in Section \ref{sec:protocol}.}), the Goldwasser, Pavlov, and Vaikuntanathan protocol works under the easier, weaker \emph{static} (rushing) adversary.

Indeed, the adaptive (rushing) adversary is much more powerful since there exists an (expected) $\Omega(t/\sqrt{n\log n})$ round lower bound for adaptive adversary  (which holds even for adaptive \emph{crash faults} and even for \emph{randomized} algorithms) shown by Bar-Joseph and Ben-Or \cite{BB98}.  This high lower bound is often cited as the reason why many works assume the static adversary \cite{Ben-Or_2006}. The best-known protocol for the adaptive adversary takes expected $O(t/\log n)$ rounds  (and tolerating up to $t < n/3$ Byzantine nodes) due to Chor and Coan \cite{CC85}, which has stood for four decades. We note that the Chor and Coan bound is only a logarithmic factor better than the best possible deterministic protocols that run in $O(t)$ rounds (e.g., \cite{Dolev_1982,garay}), but it showed, for the first time, that randomization can break the $t+1$ deterministic lower bound barrier. In this paper, we present a Byzantine agreement protocol for the adaptive adversary that significantly improves the Chor and Coan bound~\cite{CC85}.

\subsection{Model}
\label{sec:model}

This work focuses on the challenging adversarial model of an {\em adaptive}
Byzantine adversary that can choose which nodes to be Byzantine based on the protocol's execution. We assume the {\em full-information rushing} model where the Byzantine nodes can behave arbitrarily and maliciously, have knowledge about the entire state of the network at every round, including random choices made by all the nodes up to and including the current round, have unlimited computational power, and may collude among themselves. Note that in this model, cryptographic techniques do not apply. Furthermore, protocols that work under this model (like ours)
are tolerant to quantum attacks.

As is standard in all the results cited in this paper, we assume point-to-point communication between all pairs of nodes, i.e., a {\em complete} network of $n$ nodes. The Byzantine adversary can adaptively corrupt up to $t$ nodes (in total) during the protocol's execution. Our protocol can tolerate up to $t < n/3$ Byzantine nodes, which is the best possible fault tolerance in the full-information model \cite{FLM,yao}. 

Communication is {\em synchronous} and occurs via message passing, i.e., communication proceeds in discrete rounds by exchanging messages; every node can communicate directly with every other node. We assume a CONGEST model, where each node has only limited bandwidth, i.e., only $O(\log n)$ bits can be sent per edge per round.  As is standard in Byzantine agreement (see e.g., Lamport et al. \cite{Pease_1980}), we assume that the receiver of a message across an edge in $G$ knows the identity of the sender, i.e., if $u$ sends a message to $v$ across edge $(u,v)$, then $v$ knows the identity of $u$; also the message sent across an edge is delivered correctly. Thus, we assume that each node has a (unique) ID that is known to all nodes (if not, it can be learned in one round of communication).

\subsection{Our Main Result} \label{sec:result}

We present a randomized Byzantine agreement protocol under an adaptive adversary in the full information model, which significantly improves upon previous results. Our protocol (cf. Algorithm \ref{alg:ByzantineAgreementClique}) achieves Byzantine agreement with high probability and runs in $O(\min\{t^2\log n/n, t/\log n\})$ rounds in the CONGEST model (cf. Theorem \ref{thm:ByzantineAgreementClique}).\footnote{It is easy to modify our protocol so that Byzantine agreement is always reached but in $O(\min\{t^2\log n/n, t/\log n\})$ expected rounds --- cf. Section \ref{subsec:committeeBA}.} Our runtime significantly improves over the long-standing result of Chor and Coan \cite{CC85} that presented a randomized protocol running in (expected) $O(t/\log n)$ rounds (which is only a $O(\log n)$ factor better than deterministic protocols that take $O(t)$ rounds\cite{Dolev_1982,garay}).  Our protocol (like Chor and Coan) has optimal resilience as it can tolerate up to $t < n/3$ Byzantine nodes. However, our running time is significantly better than that of Chor and Coan for $t = o(n/\log^2 n)$. More precisely, our protocol's bound strictly improves on that of Chor and Coan \cite{CC85} when $t = o(n / \log^2 n)$, and (asymptotically) matches their runtime for $O(n / \log^2 n) \leq t < n/3$. For example, when $t = n^{0.75}$, our protocol takes $O(n^{0.5}\log n)$ rounds whereas Chor and Coan's bound is $O(n^{0.75}/\log n)$. The message complexity of our protocol is $O(\min\{nt^2\log n, n^2t/\log n\})$, which also improves over Chor and Coan. Furthermore, the local computation cost of our protocol is small (linear in $n$) and the amount of randomness used per node is constant. Our protocol can also terminate {\em early} as soon as agreement is reached. In particular, if there are only $q < t$ nodes that the Byzantine adversary corrupts, then the protocol will terminate in $O(\min\{q^2\log n/n, q/\log n\})$ rounds.

Our protocol's round complexity approaches the well-known lower bound of $\Omega(t/\sqrt{n \log n})$ (expected) rounds due to Bar-Joseph and Ben-Or \cite{BB98} (cf. Theorem \ref{thm:lowerBoundClique}) when $t$ approaches $\sqrt{n}$. In particular, when $t = \sqrt{n}$, it matches the lower bound within logarithmic factors, and thus, our protocol is near-optimal (up to logarithmic factors).  
We \emph{conjecture} that our protocol's round complexity is optimal (up to logarithmic factors) for all $t < n/3$.

Our protocol (cf. Section \ref{sec:protocol}) is a variant of the classic randomized protocol of Rabin \cite{Rabin_1983} as is the protocol of Chor and Coan \cite{CC85}. Rabin's protocol assumes a shared (common) coin available to all nodes (say, given by a trusted external dealer). Chor and Coan present a simple method to generate common coins by the nodes themselves without needing an external dealer. The main modification of our protocol is a more efficient way to generate shared coins using the fact that one can group nodes into committees of appropriate size to generate a common coin. Our protocol crucially makes use of the fact that even if the Byzantine adversary is adaptive, it cannot prevent the generation of a common coin (cf. Definition \ref{def:ccoin}) if the number of Byzantine nodes is at most a {\em square root} of the total number of nodes (cf. Theorem \ref{th:cc}). We show this fact using an anti-concentration inequality due to Paley and Zygmund (cf. Lemma \ref{ineq:pz}).

\subsection{Additional Related Work} \label{sec:related}

The literature on Byzantine agreement is vast (see e.g., \cite{saia,pettie,Wattenhofer_2019_Book}), and we limit ourselves to the most relevant to this work. As mentioned earlier, the best-known bound for Byzantine agreement under an adaptive adversary is a long-standing result of Chor and Coan \cite{CC85} who give a randomized protocol that finishes in expected $O(t/\log n)$ rounds and tolerates up to $t < n/3$ Byzantine nodes. We note that this protocol assumes a non-rushing adversary (though this can also be modified to work for rushing). 

The work of Augustine, Pandurangan, and Robinson \cite{Augustine_2013_PODC}  gives a protocol
for Byzantine agreement in  {\em dynamic} and {\em sparse expander} networks that can tolerate $O(\sqrt{n}/\text{polylog}(n))$ Byzantine nodes.  We note that this setting differs from the one considered here; the agreement protocol in \cite{Augustine_2013_PODC} also differs and is based on a {\em sampling majority} protocol. In the sampling majority protocol, in each round, each node samples values from two random nodes and takes the majority of its value and the two sampled values; this is shown to converge to a common value in $\text{polylog}(n)$ rounds if the number of Byzantine nodes is $O(\sqrt{n}/\text{polylog}(n))$.   We note that our common coin
protocol (Algorithm \ref{alg:complete-cc}) and the sampling majority protocol both use an anti-concentration bound in their analysis.

We note that all the above results are for {\em synchronous} networks.
There has also been work on the even harder case of adaptive adversary under {\em asynchronous} networks, where an adversary can arbitrarily delay messages sent by honest nodes (in contrast to synchronous networks, where each message
is delivered in 1 round). This model was studied starting from the seminal works of Ben-Or \cite{Ben-Or_1983} and Bracha\cite{bracha}, who gave {\em exponential} round algorithms that tolerate up to $t < n/3$ Byzantine nodes. The work of King and Saia \cite{saia} gave the first (expected) polynomial time algorithm for this model that tolerated $\epsilon n$  Byzantine nodes (for some small constant $\epsilon$). Recently, Huang, Pettie, and Zhu \cite{pettie} (see also \cite{pettie1}) improved the resilience to close to $n/3$ while running in a polynomial number of rounds. We note all the above polynomial run-time bounds are pretty large; in particular, the algorithm of \cite{pettie} takes $O(n^4)$ rounds to achieve resilience close to $n/3$.

\section{Preliminaries}
\label{sec:prelim}

\subsection{Anti-concentration Inequality on Random Variables}

We provide an anti-concentration inequality --- the Paley-Zigmund inequality --- that we will use for our common coin protocol in Section \ref{subsec:commonCoin1Round}.

\begin{lemma}[Paley-Zygmund Inequality\cite{PZ,Steele}]
\label{ineq:pz}
If $X \geq 0$ is a random variable with finite variance and  if  $0 \leq \theta \leq 1$, then
$$\Pr (X > \theta E[X]) \geq (1-\theta)^2 \frac{E[X]^2}{E[X^2]}.$$
\end{lemma}

\subsection{Byzantine Agreement Lower Bound}
\label{subsec:lowerBound}

 In $n$-node networks, a well-known result from Bar-Joseph and Ben-Or \cite{BB98} provides an $\Omega(t / \sqrt{n \log n})$ runtime lower bound for Byzantine agreement against an adaptive full-information rushing crash fault adversary that can control up to $t < n/3$ nodes. This lower bound result clearly also applies to an adaptive full-information rushing Byzantine adversary that can control up to $t < n/3$ nodes (see Theorem \ref{thm:lowerBoundClique}).

\begin{theorem}[\cite{BB98}]
\label{thm:lowerBoundClique}
Given a $n$-node network, of which at most $t < n/3$ is controlled by an adaptive full-information rushing Byzantine adversary, any algorithm solving Byzantine agreement takes $\Omega(t / \sqrt{n \log n})$ rounds.
\end{theorem}

\section{A Byzantine Agreement Protocol}
\label{sec:protocol}

We present a new randomized Byzantine agreement protocol (Algorithm \ref{alg:ByzantineAgreementClique}) in synchronous (complete) networks under an adaptive (full information rushing) adversary that improves over the longstanding result of Chor and Coan \cite{CC85}. More precisely, the runtime of Algorithm \ref{alg:ByzantineAgreementClique} strictly improves on that of Chor and Coan \cite{CC85} when $t < n / \log^2 n$, and matches their runtime for $n / \log^2 n \leq t < n/3$. Moreover, the smaller the $t$, the more significant the improvement and the runtime approaches the lower bound of Bar-Joseph and Ben-Or \cite{BB98} when $t$ approaches $\sqrt{n}$, at which point it becomes asymptotically optimal (up to a logarithmic factor).

\begin{restatable}{theorem}{firstMainResult}
\label{thm:ByzantineAgreementClique}
    Given a $n$-node network, of which at most $t < n/3$ nodes are Byzantine, Algorithm \ref{alg:ByzantineAgreementClique} solves Byzantine agreement with high probability in $O(\min\{(t^2 \log n)/n, t/ \log n \})$ rounds. Furthermore, if only $q < t$ nodes are corrupted by the adversary, then the algorithm will terminate (early) in
    $O(\min\{(q^2 \log n)/n, q/ \log n \})$ rounds.
\end{restatable}

\subsection{Common Coin Protocol}
\label{subsec:commonCoin1Round}

We first present a simple one-round common coin-generating protocol in synchronous (complete) networks of $n$ nodes that works under an adaptive full information rushing Byzantine adversary controlling up to $t = O(\sqrt{n})$  nodes. This common coin protocol is crucial to our subsequent Byzantine agreement protocol. First, we define a common coin protocol.

\begin{definition}[Common Coin]  \label{def:ccoin}
  Consider a protocol $P$ where every honest node outputs a bit
and let $\comm$ be the event that all nodes output the same bit value $b$.
If there exists  constants $\delta , \epsilon > 0$,  such that
\begin{enumerate}
\item[(A)] $\Pr(\comm) \ge \delta$, and
\item[(B)] $\epsilon \le \Pr(b = 0 \mid \comm) \le 1 - \epsilon$, 
\end{enumerate}
then we say that \emph{$P$
  implements a common coin}.
\end{definition}

\begin{algorithm}[ht]
\caption{(Single Round) Coin Flipping Protocol for (honest) node $v$}
\label{alg:complete-cc}
\begin{algorithmic}[1]
\State $X_v := Uniform(\{-1,1\})$ \Comment{Choose  1 or -1 with probability $1/2$ each}
\State Broadcast $X_v$ to all neighbors $N(v)$
\If{$\sum_{u \in N(v)} X_u \geq 0$} Return 1 \Comment{If $v$ received more 1 values than -1 values}
\Else $\;$ Return 0
\EndIf
\end{algorithmic}
\end{algorithm}

The protocol (Algorithm \ref{alg:complete-cc}) is quite simple. 
Each honest node chooses either $1$ or $-1$ with equal probability and broadcasts this value to all other nodes. Each honest node then adds up all the values received (including its value). If the value is greater than or equal to zero, the node chooses bit $1$  as the common value; otherwise, bit $0$  is chosen. Note that the Byzantine nodes can decide to send different values to different nodes, even after seeing the random choices made by the honest nodes (i.e., a \emph{rushing} adversary).  Yet, we can show the following theorem.

\begin{theorem}
\label{th:cc}
Given a $n$-node network with at most $\frac{1}{2}\sqrt{n}$ Byzantine nodes,
Algorithm \ref{alg:complete-cc} implements a common coin.
\end{theorem}

\begin{proof}
    We show that with at least some constant probability $c>0$, all nodes choose the same bit value, and that the chosen bit value is bounded away from 0 and 1 with a probability of at least $c$. 

    Algorithm \ref{alg:complete-cc} is a one-round protocol and the Byzantine adversary can choose which nodes to corrupt \emph{after} seeing the random choices of all the nodes in the first round.
    Let $G$ be the set of honest nodes (that are uncorrupted by the Byzantine nodes in the first round) and $g= |G|$. We have $g \geq n - 0.5\sqrt{n}$. 
    
    $X_v$ denotes the random choice by node $v$. For an honest node $v$, $\Pr(X_v =1) = 1/2$ and $\Pr(X_v = -1) = 1/2$.
    Let $X = \sum_{v \in G} X_v$ be the sum of the random choices of the honest nodes only. We show that $\Pr(X > 0.5 \sqrt{n}) > c$ and $\Pr(X < -0.5\sqrt{n}) > c$, where $c>0$ is a constant.

    We use the Paley-Zygmund inequality (cf. Lemma \ref{ineq:pz}) to show the above.
    We show below that $\Pr(X > 0.5 \sqrt{n}) > c$. The other inequality can be shown similarly. 
    We first note that 
    $$E[X] = E[\sum_{v \in G} X_v] = \sum_{v \in G} E[X_v] = \sum_{v \in G} (\frac{1}{2} (1)  + \frac{1}{2} (-1)) = 0.$$ 
    
    Also,
    \begin{align*}
        E[X^2] & = E[(\sum_{v \in G} X_v)^2] =  E[(\sum_{v \in G} X_v^2 + 2\sum_{u\neq v, \in G} X_uX_v)] = \sum_{v \in G}E[X_v^2] + 2\sum_{u\neq v, \in G} E[X_uX_v] \\
        &= \sum_{v \in G} (\frac{1}{2}(1) + \frac{1}{2}(1))
    + 2\sum_{u\neq v, \in G} (\frac{1}{4}(1-1-1+1)) = \sum_{v \in G} 1 + 2\sum_{u\neq v, \in G} 0 = g 
    \end{align*}
     Thus,
     \begin{align*}
        \Pr(X > \frac{1}{2} \sqrt{n}) &= \Pr(X^2 > \frac{1}{4} n) = \Pr(X^2 > (\frac{n}{4g}) g) =      \Pr(X^2 > \theta E[X^2])
    \end{align*}
     where $E[X^2] = g$ and $\theta = \frac{0.25n}{g} < 1$, since $g \geq n - 0.5\sqrt{n}$.
     Now, applying the Paley-Zygmund inequality to the non-negative random variable $X^2$, we have,
     $$\Pr(X^2 > \theta E[X^2]) \geq (1-\theta)^2 \frac{E[X^2]^2}{E[X^4]}.$$
     We showed above that $E[X^2] = g$. We next compute $E[X^4]$.
     \begin{align*}
         E[X^4] &= E[(\sum_{v \in G} X_v)^4] \\
            &= E[\sum_{v} X_v^4 + \sum_{u\neq v} X_u^3X_v +\sum_{u\neq v} X_u^2X_v^2 + \sum_{u\neq v\neq w} X_v^2X_uX_w 
             + \sum_{u\neq v\neq w \neq y} X_vX_uX_wX_y]
     \end{align*}
     Using the fact that the $X_v$s are i.i.d random variables and by linearity of expectation, we have
     \begin{align*}
         E[X^4] &= gE[X_v^4] + 4{g \choose 2} E[X_v^3X_u] + 6{g \choose 2} E[X_v^2 X_u^2] 
           + 12{ g \choose 3} E[X_v^2 X_uX_w] + 24 {g \choose 4} E[X_vX_uX_wX_y]
     \end{align*}
     In the above, we have $E[X_v^4] = 1$ and $E[X_v^2X_u^2] =1$ and the rest are 0.
     Hence, $E[X^4] = g + 3g (g-1) = 3g^2-2g$ and thus,

     \begin{align*}
         \Pr(X > 0.5 \sqrt{n}) &= \Pr(X^2 > \theta E[X^2]) \geq (1-\theta)^2 \frac{E[X^2]^2}{E[X^4]} 
         \geq (1-\theta)^2 \frac{g^2}{3g^2-2g} \geq \frac{1}{3}(1-\theta)^2
     \end{align*}
     Plugging in $\theta = \frac{0.25 n}{g}$, and since $g \geq n - 0.5 \sqrt{n} \geq n/2$, we have
     $$\Pr(X > 0.5 \sqrt{n}) \geq  \frac{1}{3}(1-\theta)^2 \geq  \frac{1}{12}.$$

     By an identical argument, it follows that $\Pr(X < - 0.5 \sqrt{n}) \geq  \frac{1}{12}.$
     
     Hence, with probability at least $1/12$, all honest nodes will have their sum evaluated to more than $0.5\sqrt{n}$
     or less than $-0.5 \sqrt{n}$. Since the adversary can corrupt at most $0.5 \sqrt{n}$ nodes, in each of these
     cases, the total sum of all values will remain positive or negative, respectively. Hence, all the honest (uncorrupted) nodes will 
     choose 1 or 0, respectively, in the above two cases.
     Thus, all honest nodes will choose a common coin with constant probability.  
\end{proof}

\paragraph{Variant Protocol with Designated Nodes for Coin Flipping.}
We will need a variant of Algorithm \ref{alg:complete-cc}, which we call Algorithm \ref{alg:complete-cc-variant}, for the agreement protocol in Subsection \ref{subsec:committeeBA}. In that variant common coin protocol, we assume some $k$ nodes are \emph{designated} --- that is, their IDs are known to all nodes --- and that among the designated nodes, there are at most $\frac{1}{2}\sqrt{k}$ Byzantine nodes. Here, these designated $k$ nodes are the only nodes to flip coins (and thus the only honest nodes to influence the common coin). Then, they broadcast their values to all $n$ nodes. Finally all $n$ nodes take the majority of their received values as their common coin value.

\begin{algorithm}[ht]
\caption{Coin-Flip: (Single Round) Coin Flipping Protocol for honest node $v$ with an additional input}
\label{alg:complete-cc-variant}
\begin{algorithmic}[1]
\State \textbf{Node set input:} $V_d$  \Comment{Set of nodes designated for flipping coins}
\If{$v \in V_d$} \Comment{Only designated nodes contribute random choices}
    \State $r_v := Uniform(\{-1,1\})$ 
    \State Broadcast $r_v$ to all neighbors $N(v)$
\EndIf
\If{$\sum_{u \in V_d} r_u \geq 0$} Return 1  
\Else $\;$ Return 0
\EndIf
\end{algorithmic}
\end{algorithm}

The correctness of Algorithm \ref{alg:complete-cc-variant} follows from that of Algorithm \ref{alg:complete-cc}.

\begin{corollary}
\label{lem:randomCoin}
Given a $n$-node network, a set of $k \geq 1$ designated nodes and at most $\frac{1}{2}\sqrt{k}$ Byzantine nodes among the designated nodes, Algorithm \ref{alg:complete-cc-variant} implements a common coin.
\end{corollary}

\subsection{Committee-Based Byzantine Agreement}
\label{subsec:committeeBA}

Now, we describe our Byzantine agreement protocol that works against an adaptive full information rushing Byzantine adversary controlling up to $t < n/3$ nodes (i.e.,  $n \geq 3t+1$).

At the start of the protocol (Algorithm \ref{alg:ByzantineAgreementClique}), each node $v$ initializes some variable $val_v$ to its binary input $input(v)$. It also initializes
a variable $decided_v$ to $False$ and $finish_v$ to $False$; these variables are used to detect termination.  Nodes group themselves into $c = \min \{ (\alpha \lceil t^2/n \rceil \log n , 3\alpha t/\log n \}$ committees (where $\alpha \geq 1$ is a well-chosen constant that depends on the analysis) of uniform size $s = n / c$ using their IDs: nodes with IDs in $\{1,\ldots,s\}$ form the first committee, nodes with IDs in $\{s+1,\ldots,2 s\}$ form the second committee and so on. (As assumed, each node knows the IDs of all of its neighbors, even that of Byzantine nodes. Additionally, note that the last committee may not be of size $s$, which we ignore in the description and the analysis due to minimal impact.)

Next, the protocol (Algorithm \ref{alg:ByzantineAgreementClique}) executed by each node $v$ consists of $c$ phases, each consisting of two (broadcast and receive) communication rounds (denoted by 1 and 2 in the message type). In the first communication round of phase $i$, each node $v$ broadcasts its $val$ and $decided$ values. Assume $finish_v = False$; otherwise, the node terminates.
Node $v$ then receives the messages sent in the first round, and checks if it received at least $n-t$ identical $val$ values $b$ (regardless of $decided$ values), in which case it will set $val_v = b$ and  $decided_v = True$. Otherwise, it will set $decided_v = False$.
Then, in the second communication round, each node $v$ broadcasts its $val$ and $decided$ values again. Depending
on the values received from all nodes, $v$ has three cases to consider. 
\begin{description}
\item [Case 1:]
If  $v$ receives the same value $b$ from at least $n-t$ nodes with $decided$ values set to $True$, then $v$ assigns that value to $val_v$; it also sets its $finish_v$ value to $True$. It will then terminate after broadcasting its value one more time to all nodes in the next phase.  
\item [Case 2:] If  $v$ receives the same value $b$ from at least $t+1$ nodes with $decided$ values set to $True$, then $v$ assigns that value to $val_v$; it also sets its $decided_v$ value to $True$.
\item [Case 3:] Otherwise (i.e., if both of the above cases do not apply), $v$  uses the coin flipping protocol described in Subsection \ref{subsec:commonCoin1Round} (Algorithm \ref{alg:complete-cc-variant}), with the designated nodes being the $s$ nodes of committee $i$. After which, $v$ assigns the common coin value to $val_v$ and sets its $decided_v$
value to $False$.
\end{description}
At the end of all $c$ phases (or if it finished earlier), node $v$ decides on $val_v$.

\begin{algorithm}[ht]
\caption{Byzantine Agreement Protocol for node $v$, tolerating up to $t$ Byzantine nodes for any $t < n/3$}
\label{alg:ByzantineAgreementClique}
\begin{algorithmic}[1]
\State \textbf{Binary input:} $input(v)$ 
\State $c = \min \{\alpha \lceil t^2/n \rceil \log n , 3 \alpha t/\log n \}$, $s = n / c$
\State $val_v := input(v)$ 
\State $decided_v := False$
\State $Finish_v = False$
\For{phase  $i = 1$ to $c$} 
    \LineComment{Round 1 of phase $i$}
    \State Broadcast $(i, 1, val_v, decided_v)$ to all nodes 
    \If{$Finish_v = True$}
    \State Go to Line \ref{line:end} \label{line:final}
    \EndIf
    \State Receive $(i,1,*,*)$ messages from all nodes
    \If{$v$ receives at least $n-t$ messages of type $(i,1,b,*)$ (with identical values $b$)} 
    \State $val_v := b$ \label{line:r1}
    \State $decided_v = True$ \label{line:dec}
    \Else
    \State $decided_v = False$ \label{line:notdec}
    \EndIf
    \State
    \LineComment{Round 2 of phase $i$}
    \State Broadcast $(i, 2, val_v, decided_v)$ to all nodes 
    \State Receive $(i,2,*,*)$ messages from all nodes
    \If{$v$ receives at least $n-t$ messages of type $(i,2,b,True)$} \label{line:ter}
    \State $val_v = b$ \label{line:r2}
    \State $decided_v = True$
    \State $Finish_v = True $  \Comment{Terminate after broadcasting once more}
    \ElsIf{$v$ receives at least $t+1$ messages of type $(i,2,b,True)$} \label{line:t+1}
    \State $val_v:= b$ \label{line:aft}
    \State $decided_v := True$
    \Else 
    \State Execute Coin-Flip (Algorithm \ref{alg:complete-cc-variant}) with input $V_d = \{u \in V \mid \lceil ID_u / s \rceil = i\}$  \label{line:coin}
    \State $val_v = $ output of Coin-Flip   \label{line:coinset}
    \State $decided_v := False$
    \EndIf
\EndFor
\State Return $val_v$ \label{line:end}
\end{algorithmic}
\end{algorithm}

\paragraph{Analysis of Algorithm \ref{alg:ByzantineAgreementClique}.} We start with the following two lemmas on the behavior of honest nodes.

\begin{lemma}
\label{obs:obsValidity}
    For any phase $i$, in Line \ref{line:r1}, if at least $n-t$ honest nodes agree on one value $b$, then all honest nodes agree on that value in Line \ref{line:r2}.
\end{lemma}

\begin{proof}
Suppose at least $n-t$ honest nodes agree on a value $b$ in Line \ref{line:r1}.
Then \emph{every} honest node $v$ will receive at least $n-t$ values of $b$  with $decided$ values set to $True$ in Line \ref{line:r2}
and hence set $val_v = b$. 
\end{proof}

\begin{lemma}
\label{obs:singleAssignedValue}
    For any phase $i$, no two honest nodes assign different values in Line \ref{line:r1}. Thus, all honest nodes with $decided$ value set to $True$ in Line \ref{line:dec} will have the same $val$ value. 
\end{lemma}

\begin{proof}
Suppose not. This means that two processors $u$ and $v$ each received at least $n-t$ identical values, which are different, and hence they respectively assigned $val_u = b$ and $val_v = 1-b$. Since there are only $t$ bad processors, and $u$ sees at least $n-t$ values of $b$, at least $n-2t$  values of $b$ are from honest processors. Since there are $n-t$ honest processors, at most $n-t -(n-2t) = t$ honest nodes can have value $1-b$. This means that $v$ can see at most $t+t = 2t$ values of $1-b$, which is a contradiction to the assumption that $v$ sees at least $n-t > 2t$ values of $1-b$. 
\end{proof}

As a result of Lemma \ref{obs:singleAssignedValue}, we can define for any phase $i$ an \emph{assigned value} $b_i$, which is the value assigned by any (honest) node in Line \ref{line:r1}. Such a node will set its $decided$ value to be $True$. We use $decided$ values to detect early termination, i.e., termination as soon as all nodes agree. (Note that, in any case, the algorithm terminates in $c$ phases.) More precisely, each node can detect agreement (and hence terminate) by making use of the number of $True$ $decided$ values it receives, as we show in the following lemma.

\begin{lemma}
\label{lem:decided}
For any phase $i \leq c-2$, if some honest node $v$ receives at least $n-t$  $decided$ values set to $True$ with corresponding identical values $b$, then it can (safely) terminate. More precisely, node $v$ terminates in phase $i+1$ and all other (honest) nodes terminate at the latest in phase $i+2$.
\end{lemma}

\begin{proof}
    First we prove   that if all honest nodes  reach agreement, then they terminate correctly. This is easy to see. Suppose all honest nodes have the same value $val =b_i$ at the beginning of a phase $i$. Then in Line \ref{line:r1}, they will set their $val = b_i$, since they will see at  least $n-t$ identical $b_i$ values. They will also set their corresponding $decided$ values to be $True$. Hence in Line \ref{line:r2}, all honest nodes will receive at least $n-t$ $True$ $decided$ values with identical $b_i$ values and hence will terminate correctly.
     
    Next, we will show that if an honest node decides to terminate after executing Line \ref{line:r2} in phase $i$, all honest nodes will agree and terminate in (at most) two additional phases (i.e., by phase $i+2$). Let some (honest) node $v$ receive at least $n-t$ $decided$ values set to $True$ (Line \ref{line:ter}), then it will terminate with the corresponding (identical) value received.  Since it received at least $n-t$ $True$ $decided$ values, at least $n-2t \geq t+1$ of these are from honest nodes.  By Lemma \ref{obs:singleAssignedValue}, all these honest nodes will have the same $val$ value, say bit $b_i$. Thus, \emph{all} honest nodes will receive at least $t+1$ $True$ $decided$ values in Line \ref{line:t+1} and will set their $val$ to be $b_i$ and set their $decided$ value to be $True$. Then, in the next phase, all (remaining) honest nodes will receive at least $n-t$  $True$ $decided$ values  (note that nodes that have decided to terminate in phase $i$ will broadcast once in phase $i+1$ and terminate --- Line \ref{line:final}) with identical $val = b_i$. Hence, all remaining honest nodes will agree on $val = b_i$ and terminate (in phase $i+2$). 
\end{proof}

Honest nodes that do not assign their respective $val$ variables to $b_i$ in  round 1 of phase $i$ (due to not receiving at least $n-t$  values of $b_i$ in round 1 of phase $i$) will set their $decided$ values to $False$ (Line \ref{line:notdec}). Still, in round 2 of phase $i$, it could be the case that some honest nodes execute Lines \ref{line:r2} or \ref{line:aft}, and set their $decided$ values to be $True$; all such nodes would set their $val=b_i$. Other honest nodes, that is those that do not receive at least $t+1$ $True$ $decided$ values with identical $val$ values, will execute the common coin flip protocol (Line \ref{line:coin}).

We say a phase $i$ is \emph{good}  if all honest nodes agree at the end of the phase. This will happen if the coin flip value is the same for all honest nodes executing the flip (i.e., the Coin Flip protocol  succeeded in generating a common coin) and the common coin value agrees with the $b_i$ value of those honest nodes with $True$ $decided$ values (set in Line \ref{line:aft}). In the following lemma, we show that for any phase $i$, if few enough Byzantine nodes are part of the $i$th committee (executing the coin flip of phase $i$), then the phase is good.

\begin{lemma}
\label{lem:goodPhase}
    For any phase $i$, if less than $\frac{1}{2} \sqrt{s}$ Byzantine nodes are part of committee $i$, then the phase is good --- that is, all honest nodes agree on the same value ($val$) --- with constant probability.
\end{lemma}

\begin{proof}
    By Lemma \ref{obs:singleAssignedValue}, all honest nodes with $decided$ set to $True$ after the first round all have $val = b_i$. Hence, any honest node that sets its $val$ value according to cases 1 and 2 (Lines \ref{line:r2} and \ref{line:aft}) in the second round set $val$ to $b_i$. Next, it suffices to show that the other honest nodes (setting $val$ according to case 3 (Line \ref{line:coinset})  in the second round) set $val$ to $b_i$ with constant probability to prove that the phase is good with constant probability. Note that even if the assigned value $b_i$ is chosen arbitrarily by a rushing Byzantine adversary in the first round of phase $i$, it must be chosen independently of the honest nodes' random choices in that phase's second round. This includes, in particular, the random choices that decide the coin flip. Thus, by Corollary \ref{lem:randomCoin} (and the definition of a common coin), there is a constant probability that the common coin outputs $b_i$. In which case, honest nodes following case 3 set $val = b_i$, and the phase is good.
\end{proof}

Now, we can prove this section's main result --- Theorem \ref{thm:ByzantineAgreementClique}. 

\firstMainResult*

\begin{proof}
    The round complexity follows directly from the description of Algorithm \ref{alg:ByzantineAgreementClique}. Moreover, given that $t < n/3$, Lemma \ref{obs:obsValidity} implies that Algorithm \ref{alg:ByzantineAgreementClique} satisfies the validity condition. Next, we show that Algorithm \ref{alg:ByzantineAgreementClique} satisfies the agreement condition (with high probability). For the analysis, we consider two different regimes --- $t \leq n / \log^2 n$ and $t > n/ \log^2 n$ --- and prove there is at least one good phase with high probability separately for both. This implies Algorithm \ref{alg:ByzantineAgreementClique} satisfies the agreement condition with high probability.
    
    For the first regime, a straightforward counting argument implies that there are at least $0.5 \sqrt{s}$ Byzantine nodes in at most $2 t / \sqrt{s} = 2 t \sqrt{c / n}$ committees. Let us lower bound the number of committees with strictly less than $0.5 \sqrt{s}$ Byzantine nodes. 
    Note that $c = \alpha \lceil t^2 / n \rceil \log n$ in this regime ($t \leq n / \log^2 n$). Hence, $2 t \sqrt{c / n} \leq 2 t \sqrt{\alpha} (t / \sqrt{n} + 1) \sqrt{(\log n)/ n}$ (since $\sqrt{t^2/n+1} \leq t/\sqrt{n}+1$, as the square root function is subadditive over the positive real numbers). If $t \leq \sqrt{n}$, then $c - 2 t \sqrt{c / n} \geq \alpha \log n - 4 \sqrt{\alpha \log n}$. Else, $c - 2 t \sqrt{c / n} \geq \alpha (t^2 / n) \log n - 4 \sqrt{\alpha} (t^2/n) \sqrt{ \log n }$.
    In both cases, for any constant $\gamma \geq 1$ (where this constant decides the desired polynomial term in the failure probability), choosing $\alpha$ such that $\alpha - 4 \sqrt{\alpha} \geq \gamma$ implies that there are $c - 2 t \sqrt{c / n} \geq \gamma \log n$ committees with at most $0.5 \sqrt{s}$ Byzantine nodes. Thus, by Lemma \ref{lem:goodPhase}, the corresponding $\gamma \log n$ phases are each good with constant probability, say $p$. Since the event that a phase is good is independent of the other phases, we get that at least one phase is good with probability at least $(1- (1-p)^{\gamma \log n})$. (Recall that $1-p$ is constant.) In other words, at least one phase is good and thus all honest nodes agree, with high probability (for $\alpha$ chosen accordingly).

    Now, let us consider the second regime: $t > n/ \log^2 n$. There are $t < n/3$ Byzantine nodes; hence, there are $3 \alpha t/\log n$ committees of size at least $(1/\alpha) \log n$ each. A straightforward counting argument implies that there can be at least $(1/2\alpha) \log n$ Byzantine nodes in at most $2\alpha t/ \log n$ committees. In other words, there are $a \geq \alpha t / \log n \geq \alpha n / \log^3 n$ committees with strictly less than $(\log n)/(2 \alpha)$ Byzantine nodes each (i.e., strictly more than half of the committee is honest). Now, consider any of the corresponding phases. The phase is good if (but not only if) at least $(\log n)/(2 \alpha) + 1$ honest nodes in that committee, during the coin flip, all choose 1 (respectively, -1) when the phase's assigned value is 1 (resp., 0). (Note that in this phase, only the nodes of that committee --- or designated nodes --- flip coins and contribute to the coin flip; messages from byzantine nodes not in the committee are ignored by all honest nodes.) This happens with probability $p \geq 1 / 2^{(\log n)/(2 \alpha) + 1} \geq  1 / (2 \sqrt{n})$. Since the event that a phase is good is independent of the other phases, we get that none of the phases in which there are strictly less than $(\log n)/(2 \alpha)$ Byzantine nodes are good with probability at most 
    \begin{align*}
        (1-p)^{a} & \leq (1-1 / (2 \sqrt{n}))^{n / \log^3 n} \\
                   & \leq \exp(-(1 / (2 \sqrt{n})) \cdot (\alpha n / \log^3 n)) \\
                   & \leq \exp(-\sqrt{n} / (2\log^3 n)  )  
    \end{align*}
    where the second inequality uses $1-x \leq \exp(-x)$ for any $x \in \mathbbm{R}$, and the third from $\alpha \geq 1$. Hence, for large enough $n$, at least one phase is good and all honest nodes agree, with high probability .

     Finally, due to Lemma \ref{lem:decided}, if nodes agree in any phase $i < c-2$, then they terminate within the next two phases. This implies that if the  actual number of nodes that the Byzantine adversary corrupts is $q < t$, then Algorithm 3 will finish (early) in $O(\min\{(q^2 \log n)/n, q/ \log n \})$ rounds.
\end{proof}

\noindent {\bf Las Vegas Byzantine Agreement.}
It is easy to state and prove Theorem \ref{thm:ByzantineAgreementClique} so that the guarantee is \emph{Las Vegas}, i.e., Byzantine agreement is  always reached  in $O(\min\{t^2\log n/n, t/\log n\})$ expected rounds. This is accomplished by
 modifying Algorithm \ref{alg:ByzantineAgreementClique} slightly. We allow the algorithm to simply keep iterating through the committees, starting over once the $c$th committee is reached (instead of terminating), and the early termination component of the algorithm will ensure eventual termination with the above-mentioned expected round complexity.

\section{Conclusion and Open Problems}
\label{sec:conc}

We presented a simple Byzantine agreement protocol that  improves over the runtime of the long-standing bound of Chor and Coan \cite{CC85}. Our protocol runs in $\tilde{O}(t^2/n)$ rounds\footnote{$\tilde{O}$ and $\tilde{\Omega}$ notations hide a logarithmic factor.} as compared to the $O(t/\log n)$ rounds of Chor and Coan. Thus, in the regime of (say) $t=O(n^{1-\epsilon})$, for \emph{any} constant $\epsilon$, our protocol is significantly faster.
Both protocols are randomized, have optimal $t < n/3$ resilience, and improve over the deterministic lower bound of $t+1$ rounds. But as $t$ becomes smaller, our protocol's improvement over Chor and Coan grows more and more significant. In fact, our protocol approaches the Bar-Joseph and Ben-Or lower bound of $\tilde{\Omega}(t/\sqrt{n})$ rounds when $t$ approaches $\sqrt{n}$, at which point it becomes asymptotically optimal (up to logarithmic factors).

There are two important open problems. The first is to close the gap between the upper and lower bounds. In particular, we conjecture that our protocol is near-optimal, i.e., $\tilde{\Omega}(t^2/n)$ is a lower bound on the round complexity of Byzantine agreement under an adaptive adversary in the synchronous full-information model.  The second is (possibly) improving our protocol's communication (message) complexity. Our protocol has message complexity $\tilde{O}(nt^2)$  which improves over Chor and Coan \cite{CC85}, but is still a $\tilde{O}(t)$ factor away from the best-known lower bound of $\Omega(nt)$ \cite{HadzilacosH93}.

\begin{acks}
Gopal Pandurangan was supported in part by ARO grant W911NF-231-0191 and NSF grant CCF-2402837.
\end{acks}

\bibliographystyle{plainurl}
\bibliography{reference}

\end{document}